\documentclass[conference,final,twocolumn]{IEEEtran}
\usepackage{mathrsfs}
\usepackage{algorithm}
\usepackage{algorithmic}
\usepackage{graphicx}
\usepackage{cite}
\usepackage{citesort}
\usepackage{color}
\usepackage{psfrag}
\usepackage{subfigure}
\usepackage{amssymb}
\usepackage{epsfig}
\usepackage{pifont}
\usepackage{amsmath}
\usepackage{array}
\usepackage{multicol}
\usepackage{multirow}
\usepackage{pifont}
\usepackage{indentfirst} 
\usepackage{amsfonts}
\usepackage{fancyhdr}
\usepackage{amscd}
\usepackage{bm}
\usepackage{fancyhdr}

\newtheorem{proposition}{Proposition}

\newtheorem{remark}{Remark}

\IEEEoverridecommandlockouts
\begin{document}

\title{Energy-saving Resource Allocation by Exploiting the Context Information}



\author{
\authorblockN{\large{Chuting Yao and Chenyang Yang}}
\vspace{0.2cm}
\authorblockA{Beihang University, China\\
Email: \{ctyao, cyyang\}@buaa.edu.cn } \and
\authorblockN{\large{Zixiang Xiong}} \vspace{0.2cm}
\authorblockA{ Texas A\&M University, US/Beihang University, China \\
Email: zx@ece.tamu.edu}
\thanks{This work was supported by National Natural
Science Foundation of China under Grant 61120106002 and National Basic
Research Program of China under Grant 2012CB316003.}
}

%

%
\maketitle

\begin{abstract}
Improving energy efficiency of wireless systems by exploiting the context
information has received  attention recently as the smart phone market
keeps expanding. In this paper, we devise energy-saving resource
allocation policy for multiple base stations serving non-real-time traffic
by exploiting three levels of context information, where the background
traffic is assumed to occupy partial resources. Based on the solution from
a total energy minimization problem with perfect future information, a
context-aware BS sleeping, scheduling and power allocation policy is
proposed by estimating the required future information with three levels
of context information. Simulation results show that our policy provides
significant gains over those without exploiting any context information.
Moreover, it is seen that different levels of context information play
different roles in saving energy and reducing outage in transmission.
\end{abstract}
\section{Introduction}

Energy efficiency (EE) is one of the major design goal for five-generation
(5G) mobile communication systems.

Recently, improving EE by exploiting context information has drawn
significant attention as the smart phone popularizes. Context information
can be classified into application (e.g., quality of service (QoS)), network
(e.g., congestion status), user (e.g., location or mobility pattern), and
device levels \cite{Choongul2012concept}. The user level context information
can be exploited for predicting transmission rate to assist resource
allocation in future time. For example, assuming perfect rate prediction,
the transmission time or total power was minimized to save power in
\cite{Abou2013Predictive,abou2014toward}, and a scheduling and antenna
closing strategy was proposed to improve EE in
\cite{Draxler2014Anticipatory}. In
\cite{Abou2013Predictive,abou2014toward,Draxler2014Anticipatory}, the
knowledge of the error-free rate prediction plays important role in power
saving, and the base station (BS) is assumed only serving one kind of
traffic (e.g., multimedia streaming).

In real-world systems, the predicted channel gains is never error-free.
Moreover, a BS may serve multiple classes of traffic, e.g., on demand
real-time (RT) traffic such as voice and video that are with high priority,
and non-real-time (NRT) traffic such as pre-subscribed file downloading and
emails. When there are errors in the predicted future channel gains and only
partial resource is available at the BS due to serving other traffic, the
predicted transmission rates are inevitably inaccurate, which will lead to
the EE reduction.

In this paper, we exploit context information from three levels,
application, network and user levels, for energy saving. We design
energy-minimizing resource allocation for the user with NRT service when
some background traffic such as RT service occupies partial transmission
resource. The NRT traffic is often modelled as best effort traffic  and is
served  right after the requests arrive if resource is available. As a
result, the spectral efficient or energy efficient optimization for this
kind of traffic is usually to maximize the capacity or maximize the EE of
the network. Nonetheless, if we know the expected time to accomplish the
service either by user subscription or by user behavior prediction, e.g.,
for a user-subscribed file downloading with a desired deadline, it is
possible to further save energy with long term resource allocation by
waiting for better channel condition and network status.

To demonstrate how such an application level context information can be
exploited to save energy, we artificially model the NRT service as
transmitting a given amount of data in a long duration but with a deadline,
which can reflect some emerging mobile video traffic based on subscription
or popular file prefetching based on user interest prediction. By
formulating and solving a total energy minimizing problem with perfect
future information, the optimal resource allocation policy is obtained,
which provides the intuition on how to exploit  context information as well
as a performance upper bound. By estimating the required future information
with three levels of context information, a BS sleeping, scheduling and
power allocation policy is then proposed. Simulation results show that the
proposed policy provides substantial gain in saving energy over those not
exploiting the context information. Moreover, the context information from
the application and users helps save energy and that from network helps
reduce outage probability.

\vspace{-1mm}

\section{System Model}\vspace{-1mm}
\subsection{Traffic Model and System Model}

Consider a downlink multicell system. $M$ BSs each equipped with $N_t$
antennas serve a single antenna user moving across the cells who demands NRT
traffic  over multiple time slots, where each BS also serves randomly
arrived RT traffic requests. To capture the essence of the problem, we only
consider the simple scenario of one user with NRT traffic, and leave the
issues for multi-user senario in future work. The maximal transmit power of
each BS is $P_{\max}$, and the maximal available bandwidth is $W_{\max}$.

The RT traffic has QoS provision and needs to be served immediately after
the requests arrive the BS. To ensure the QoS of the RT traffic, a given
potion of the resources need to be reserved for each request
\cite{das2003framework}. In practice, since the requests arrive randomly,
the resource used by these requests is time-varying. The transmit power and
bandwidth occupied by the RT traffic of the  $i$th BS in the  $t$th time
slot are denoted as $p^t_{i,\rm RT}$ and $W^t_{i,\rm RT}$, respectively.

The NRT traffic can be modelled as a problem to convey a given number of
bits $B$ within a duration of $T$ time slots, where the duration is much
longer than the transmission time (say, on the order of minutes or even
hours). If the $B$ bits are not reliably transmitted within the duration, an
outage occurs.

Since the NRT traffic is less urgent, the NRT user can use the remaining
transmit power and bandwidth of the RT traffic. Suppose that the request of
the NRT user arrives at the first time slot. We can select proper time slots
with proper transmit power to serve the user in order to save energy. Denote
${\bf m}_i =[m^1_i,\ldots,m^T_i]^H$ as the indicator of scheduling status
for the NRT user by the $i$th  BS, where $m^t_i \in \{0,1\}$. When
$m^t_i=1$, the user is scheduled in the $t$th time slot  by the $i$th BS.
When $m^t_i=0$, the user is not scheduled by the $i$th BS. For simplicity,
we assume that the user is only served by the closest BS in each time slot,
which allows us exploiting the user level context information as will be
clear later.

For the user with NRT traffic, the received signal in the $t$th time slot
can be expressed as
\begin{equation}\label{E:signal}
    y^t = m^t_i \sqrt{\alpha^t} ({\bf h}^t)^H {\bf
    w}^t\sqrt{p^t} x^t
    + n^t,
\end{equation}
where $x^t$ is the transmit symbol for the user in the $t$th time slot with
$\mathbb{E}\{|x^t|^2\}=1$ and $\mathbb E\{\cdot\}$ represents expectation,
$p^t$ is the transmit power, ${\bf{w}}^{t} \in \mathbb{C}^{N_t \times 1}$ is
the beamforming vector, $\alpha^t$ is the large-scale fading gain including
path loss and shadowing between the user and the closest BS, ${\bf{h}}^{t}
\in \mathbb{C}^{N_t \times 1}$ is the independent and identically
distributed (i.i.d.) small scale Rayleigh fading channel vector, and $n^t$
is the noise with variance $\sigma^2$. Since the NRT user is scheduled only
by one BS in each time slot, maximum ratio transmission is optimal, i.e.,
${\bf{w}}^{t} = {\bf{h}}^{t} /\|{\bf{h}}^{t} \|$, where $\|\cdot\|$ denotes
Euclidean norm.

In the $t$th time slot, the achievable rate of the user is
\begin{equation}\label{E:Rate}
R^t  = m^t_i W^t_i \log_2(1+g^t p^t),
\end{equation}
where $W^t_i \triangleq W_{\max}-m^t_i W^t_{i,\rm RT}$ is the available
bandwidth for the NRT user in the  $t$th time slot, $g^t \triangleq{
\alpha^t\|{\bf h}^t\|^2}/{(G\sigma^2)}$ is the equivalent channel gain, and
$G$ is the signal-to-noise ratio (SNR) gap that reflects the gap between the
capacity-achieving and practical modulation and coding selection (MCS)
transmit policies \cite{cioffi1991multicarrier}. The SNR gap depends on the
practically-used MCS  and the targeted error probability.

\vspace{-1mm}
\subsection{Power Model}
Because we strive to save energy by long term resource allocation for the
NRT traffic, the RT service is called background traffic in the rest of the
paper.

In the $t$th  time slot, the power consumed by the background traffic of BSs
(referred to as \emph{basic power} in the sequel) can be modeled as
\cite{Auer2011}
\begin{equation}\label{E:PowerModel_DS}
P^t_{\rm B}= \textstyle\sum_{i=1}^M\big(\tfrac{1}{\xi} p^t_{i,\rm RT} + {\rm Sign}(p^t_{i,\rm RT}) (P_{\rm
act} -  P_{\rm sle})+P_{\rm sle}\big),
\end{equation}
where $\xi$ is the power amplifier (PA) efficiency, $P_{\rm act}$ and
$P_{\rm sle}$ are the circuit power consumption when the BS is in active and
sleeping mode, respectively, and ${\rm Sign}(x) = \left\{
{\begin{array}{*{20}{c}}
{1,}&{x>0}\\
{0,}&{\rm else}
\end{array}} \right.$.

Further considering the power consumed by the NRT traffic, the total power
consumption at the BSs  in the $t$th  time slot is,\vspace{-1mm}
\begin{align}\label{E:PowerModel}\nonumber
P^t
=&\textstyle \sum_{i=1}^M \tfrac{1}{\xi}(p^t_{i,\rm RT}+ m^t_i p^t)+\\
 & \textstyle \sum_{i=1}^M\big({\rm Sign}(p^t_{i,\rm RT}+ m^t_i p^t)(P_{\rm act}-P_{\rm sle})+P_{\rm sle}\big).
\end{align}

\section{Resource Allocation with Future Information}\label{Sec:Opt_Single}
In this section, we optimize scheduling and power allocation  for the NRT
user to minimize the total energy consumed in the $T$ time slots. To provide
a performance upper bound and gain insights on how to design a viable
resource allocation, we assume that when optimizing at the $t$th time slot
all the information during $T$ time slots are perfectly known, which include
equivalent channel gain $g^t$ of the NRT user, the bandwidth $W^t_i$ and
transmit power $p^t_{i,\rm RT}$ occupied by the background traffic at the
$i$th BS in all time slots, $t=1,\ldots,T$.

For the user with NRT traffic, the requested $B$ bits need to be transmitted
before the deadline. The resource allocation problem to minimize the overall
energy consumption in $T$ time slots under this constraint can be formulated
as follows,\vspace{-2mm}
\begin{subequations}\label{P:2}
\begin{align}\label{P:0}
  \min_{\bf p,\bf m}~&\textstyle\sum\nolimits_{t=1}^T P^t\Delta_{T}
  \\ \label{P:1b}
  s.t. ~& \textstyle\sum\nolimits_{t=1}^T m^t_i W^t_i \log_2 (1+g^t p^t) \Delta_{T}=
  B,
\\ \label{P:1a}
&p^t\geq0,m^t_ip^t+p^t_{i,\rm RT}\leq P_{\max},\\ \nonumber
  &t=1,\ldots,T, i =1,\ldots, M,
\end{align}
\end{subequations}
where ${\bf p} = [p^1,\ldots,p^T]^H$ is the power allocation of the user
during all $T$ time slots, ${\bf m} = [{\bf m}_1, \ldots, {\bf m}_M]$ is the
scheduling status matrix of all BSs during all  time slots, and $\Delta_{T}$
is the duration of each time slot. \eqref{P:1b} is the transmission
constraint of the NRT user, \eqref{P:1a} is the power constraint of the BSs.

Since the NRT user is only accessed by the closest BS in each time slot, we
omit subscript $i$ of $m^t_i,p^t_{i,\rm RT},W^t_i$ for notation simplicity.
Then, $m^t,p^t_{\rm RT},W^t$ represent  the indicator of whether the NRT
user is scheduled by its closest BS, the power allocated to the user,  and
the available bandwidth of the closest BS in the $t$th time slot,
respectively.

According to whether the closest BS is occupied or free from the background
traffic, we divide the $T$ time slots into busy time and idle time. Denote
$\mathcal T_{\rm oc}= \{t|p^t_{\rm RT}>0\}$  with cardinality $T_{\rm oc}$
as the index set of the busy time slots, and denote $\mathcal T_{\rm id} =
\{t|p^t_{\rm RT}=0\}$  with cardinality $T_{\rm id} = T-T_{\rm oc}$ as the
index set of the idle time slots. Then, $\mathcal T_{\rm id}$ is the
complementary set of $\mathcal T_{\rm oc}$, as illustrated in Fig.
\ref{F:01}.

\begin{figure}
\centering
\includegraphics[width=0.45\textwidth]{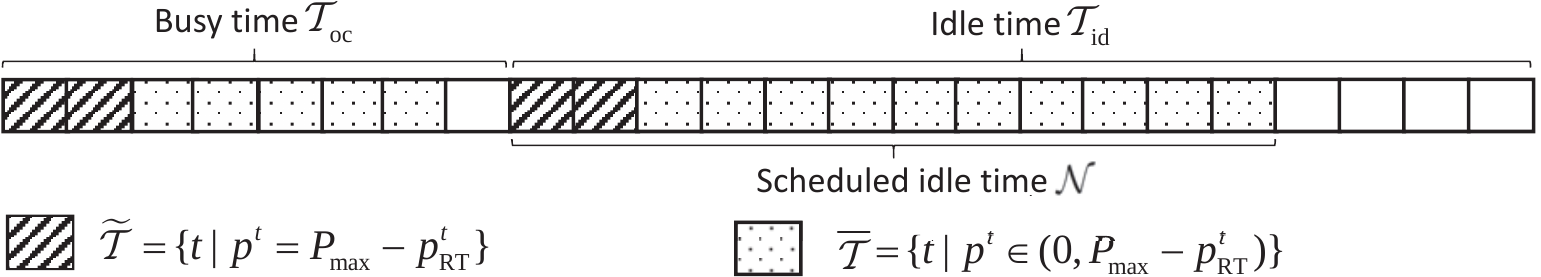}
\caption{Illustration of several sets of the time slots. To help understand,
we have re-ordered the $T$ time slots, where the equivalent channel gains in idle time slots are in a
descending order. }
\label{F:01}
\end{figure}

To fully utilize the transmit power and bandwidth at the BS, the user with
NRT traffic may be served in both kinds of time slots. Then, from
\eqref{E:PowerModel_DS} and \eqref{E:PowerModel}, the overall energy
consumption in the $T$ time slots can be rewritten as follows,\vspace{-2mm}
\begin{align}\nonumber
\sum_{t =1}^T \!\!P^t\Delta_{T}\!\!
=\!\!\!\!\underbrace{\!\!\!\!\sum_{t =1}^T\!\! P^t_{\rm B}}_{\rm Basic~Power}\!\!\!\!\!\!\Delta_{T}\!\!+\!\!
\!\!\underbrace{\!\!\sum_{t =1}^T\!\!\tfrac{1}{\xi} m^t p^t}_{\rm Transmit~Power}\!\!\!\!\!\!\Delta_{T} \!\!+
\!\!\underbrace{\!\!\sum_{t\in\mathcal T_{\rm id}}\!\! m^t(P_{\rm act}-P_{\rm sle})}_{\rm Circuit~Power}\Delta_{T},
\end{align}
which includes the basic power consumed for background traffic, transmit
power for NRT traffic, and extra circuit power for the NRT traffic in idle
time slots. Since the basic power is not affected by the resource allocation
for the NRT traffic, minimizing the overall energy consumption is equivalent
to minimizing the sum of the second and third terms.


Denote ${\cal N} = \{t|m^t=1,t\in {\cal T}_{\rm id}\}$ as the index set of
scheduled idle time slots for the NRT user. Then, $N \triangleq
\sum_{t\in\mathcal T_{\rm id}} m^t$ is the number of the scheduled idle time
slots. During the non-scheduled idle time slots, the transmit power
allocated to the NRT user is zero, i.e., $p^t = 0,m^t = 0,t\in{\cal T}_{\rm
id}- {\cal N}$, and the BS is turned into sleeping mode. Consequently,
problem \eqref{P:2} can be equivalently transformed into the following
problem
\begin{subequations}\label{P-T:2}
\begin{align}\label{P-T:2a}
  \min_{{\bf p},N}~&\textstyle\sum\nolimits_{t\in {\cal T}_{\rm oc}\cup {\cal N}}\frac{1}{\xi}  p^t\Delta_{T}+N (P_{\rm act}-P_{\rm sle})\Delta_{T}
  \\ \label{E:P-T2b1}
  s.t. ~&\textstyle\sum\nolimits_{t\in {\cal T}_{\rm oc}\cup {\cal N}} W^t \log_2(1+g^t p^t)\Delta_{T} = B,\\ \label{E:P-T2b2}
&p^t\geq0,p^t+p^t_{\rm RT}\leq P_{\max},{{t\in {\cal T}_{\rm oc}\cup {\cal N}}},  {\cal N} \subseteq {\cal T}_{\rm id}.
\end{align}
\end{subequations}

After $p^t, t=1,\cdots, T$ and $N$ are found, the solution of $m^t$ can be
obtained as will be stated later.

To solve problem \eqref{P-T:2}, we first optimize $p^t$ with given $N$ and
then exhaustively search the optimal value of $N$  from $T_{\rm id}$ to $1$
that minimizes the total energy consumption.

When $N$ is given, the transmit power can be obtained from the following
standard power allocation problem
\begin{align}\label{P:powerallocation}
\min_{{\bf p}}~&\textstyle\sum\nolimits_{t\in {\cal T}_{\rm oc}\cup {\cal N}}\frac{1}{\xi}
p^t,\\ \nonumber
s.t. ~&\eqref{E:P-T2b1}, \eqref{E:P-T2b2}.
\end{align}
whose optimal solution can be found by water-filling, i.e., \vspace{-1mm}
\begin{equation}\label{E:tranmission power_single}
  p^t = \textstyle\big(\xi \nu  W^t \log_2 e- \frac{1}{g^t}\big)_{0}^{P_{\rm
  max}-p^t_{\rm RT}}, t\in {\mathcal T}_{\rm oc}\cup \mathcal N,
\end{equation}
where $(\cdot)_0^{P_{\max}-p^t_{\rm RT}}$ represents $0 \leq p^t \leq
P_{\max}-p^t_{\rm RT}$, and $\nu$ is Lagrange multiplier.

Denote $\tilde {\mathcal T} = \{t|p^t=P_{\rm max}-p^t_{\rm RT} \}$ as the
index set of the time slots when the NRT user is served with all remaining
transmit power $P_{\rm max}-p^t_{\rm RT}$, and $ \bar {\mathcal
T}=\{t|p^t\in (0,P_{\rm max}-p^t_{\rm RT})\}$ as the index set of the time
slots when $p^t<P_{\rm max}-p^t_{\rm RT}$. Then, the Lagrange multiplier can
be expressed as follows,\vspace{-2mm}
\begin{equation}\label{E:waterfilling_Single}
\textstyle\nu = \frac{2^\frac{B-\sum_{t\in \tilde{ \mathcal T}}W^t \log_2(1+g^t
p^t)\Delta_T }{\sum_{t\in \bar {\mathcal T}} W^t\Delta_T}\ln 2}{\xi\prod_{t\in \bar {\mathcal
T}}
(W^t g^t)^\frac{W^t}{\sum_{t\in \bar {\mathcal T}} W^t}}.
\end{equation}

With \eqref{E:tranmission power_single}, \eqref{P-T:2a} becomes a function
of $N$. Then, the optimal number of the scheduled idle time slots can be
found from $N^* = \arg \min_{N} \sum\nolimits_{t\in\mathcal T_{\rm oc}\cup
\mathcal N}\frac{1}{\xi} p^t\Delta_T+N (P_{\rm act}-P_{\rm sle})\Delta_T$ by
exhaustive searching. With $N^*$, the optimal power allocation during all
the $T$ time slots can be obtained as\vspace{-2mm}
\begin{align}\label{E:powerallocation_perfect}
  \textstyle  p^{t*}  =\!\!\left\{\!\!\!\! {\begin{array}{*{20}{c}}
{\big(\xi \nu ^* W^t \log_2 e- \frac{1}{g^t}\big)_{0}^{P_{\rm
  max}-p^t_{\rm RT}},}&{\!\!\!\!t\in {\mathcal T}_{\rm oc}\cup \mathcal N^{*} }\\
{0,}&{\!\!\!\!t\in{\cal T}_{\rm id}-\mathcal N^{*} }
\end{array}} \right.
\end{align}
where $\nu^*$ is the Lagrange multiplier of problem
\eqref{P:powerallocation} for $N^*$, and ${\cal N}^*$ is index set of the
scheduled $N^*$ idle time slots.

From the water-filling structure of power allocation in
\eqref{E:powerallocation_perfect}, we can see that ${\cal N}^*$ is the $N^*$
idle time slots with highest equivalent channel gains. Consequently, ${\cal
N}^*$ can be obtained as  $\{t|g^t\geq g^*_{\rm th},t\in {\cal T}_{\rm
id}\}$, where the equivalent channel gains of $N^*$ idle time slots exceed
the threshold $g^*_{\rm th}$.

Then, the optimal power allocation to the NRT user can be accomplished with
the following \emph{PC-Algorithm}.
\begin{itemize}
  \item When the BS is in busy time, the allocated power is
      $p^{t*}=\big(\xi \nu ^* W^t \log_2 e- \frac{1}{g^t}\big)_{0}^{P_{\rm
      max}-p^t_{\rm RT}}$.
  \item When the BS is in idle time (i.e., $p^t_{\rm RT}=0$, hence, $W^t
      =W_{\max}$), if the channel condition is good enough such that
      $g^t\geq g^*_{\rm th}$, the allocated power is $p^{t*}=\big(\xi \nu
      ^* W_{\rm max} \log_2 e- \frac{1}{g^t}\big)_{0}^{P_{\rm max}}$.
      Otherwise, no power is allocated.
\end{itemize}

With $p^{t*}$, the optimal scheduling indicator for the NRT user can be
obtained as $m^{t*} = {\rm Sign}(p^{t*}),t = 1,\ldots,T$. When $p^{t*} =0$,
$t\in{\cal T}_{\rm id}- {\cal N}$,  the BS turns into sleeping mode.


\begin{remark} The power allocated in the $t$th time slot
depends on the information in this time slot including the equivalent
channel gain $g^t$, available bandwidth $W^t$ and transmit power $p^t_{\rm
RT}$, as well as the information in other time slots implicitly included in
the Lagrange multiplier $\nu^*$ and threshold $g^*_{\rm th}$. It is
noteworthy that the power allocation among the $T$ time slots shares the
same Lagrange multiplier $\nu^*$ and threshold $g^*_{\rm th}$.
%
\end{remark}
\vspace{-1mm}

\section{Resource Allocation with Context Information}
\emph{PC-Algorithm} provided in last section is only viable if the
information during the $T$ time slots can be estimated. In practice,
however, when optimizing the resource allocation for the NRT users in the
$t$th time slot, only current information $g^t$, $W^t$ and $p^t_{\rm RT}$
are known. The future information after $t$th time slot is hard to predict
accurately, especially the small scale fading gains and excess resources
available for the NRT traffic at each BS. Fortunately, only the Lagrange
multiplier $\nu^*$ and threshold $g^*_{\rm th}$ depend on the information
from future time slots. Inspired by this observation, in the sequel we
estimate these two parameters with the help of context information.

\subsection{Estimation of Lagrange Multiplier $\nu^*$} Recall that $\nu^*$ is the optimal Lagrange multiplier obtained from problem
\eqref{P:powerallocation} with given $N^*$. From the procedure of finding
the optimal resource allocation with future information, we know that to
estimate $\nu^*$ we need to estimate $N^*$ and ${\cal T}_{\rm oc}$, as well
as the values of $g^t$, $W^t$ and $p^t_{\rm RT}$ during $T$ time slots.
We exploit the network and user level context information to help estimate
these values and finally estimate $\nu^*$.

Busy time ${\cal T}_{\rm oc}$ (or equivalently its complementary set ${\cal
T}_{\rm id}$) reflects the resource utilization status, which can be
estimated from \emph{network level context information}, i.e., the average
number of busy time slots $\hat T_{\rm oc}$ from the statistics of traffic
load in the past in real-world systems. With the estimated value $\hat
T_{\rm oc}$, we can estimate the busy time slots set $\hat {\cal T}_{\rm
oc}$ by randomly choosing $\hat T_{\rm oc}$ elements from $\{1,\ldots,T\}$.
Because $\hat {\cal T}_{\rm oc}$ is used for estimating $\nu^*$ rather than
truly applied for allocating resource, such a random guess causes little
performance loss as will be shown in simulations later. Then, the estimated
index set of idle time slots $\hat {\cal T}_{\rm id}$ is the complementary
set of $\hat {\cal T}_{\rm oc} $ with cardinality $\hat T_{\rm id} = T -\hat
T_{\rm oc}$.

Equivalent channel gain $g^t$ includes both large and small channel fading
gains, which can be estimated from \emph{user level context information},
i.e., the predicted trajectory of a mobile NRT user in a given period and
the corresponding average speed. Since the smart phones popularize and the
human behavior becomes predictable, it is possible to estimate the
trajectory in a given period, which can be divided into different parts
according to the average speeds, say when a user takes different kinds of
transportation. Then, the large scale fading gains can be obtained from the
long term location prediction \cite{froehlich2008route}.
Based on the \emph{user level context information}, the user locations can
be roughly predicted, and finally the large scale fading gains of the user
during the $T$ time slots can be estimated as $\hat
\alpha^1,\ldots,\hat\alpha^T$. Considering that ${\mathbb E}\{\|{\bf
h}^t\|^2\} = N_t$, the equivalent channel gains can be estimated as $\hat
g^t, t=1, \ldots, T$, where $ \hat g^t = N_t\hat\alpha^t/(G\sigma^2)$.

Bandwidth $W_{\max}-W^t$ and transmit power $p^t_{\rm RT}$ are the resources
occupied by RT traffic during busy time slots, which are hard to predict due
to the randomness of the arrived request. To ensure the $B$ bits of the NRT
traffic being able to be conveyed within $T$ time slots, we simply assume
the worst case where no bandwidth and power are left for the NRT traffic
when a BS is busy with RT traffic based on the estimated busy time set
$\hat{\cal T}_{\rm oc}$ from \emph{network level context information}.

The variable $N^*$ is the optimal number of the scheduled idle time slots.
Using the context information to problem \eqref{P-T:2}, $N^*$ can be
estimated from the following problem \vspace{-1mm}
\begin{subequations}\label{P-T:2-est}
\begin{align}\label{P-T:2a-est}
  \min_{\hat{\bf p},\hat N}~&\textstyle\sum\nolimits_{t\in \hat {\cal N}}\frac{1}{\xi}  \hat p^t
  \Delta_T+\hat N (P_{\rm act}-P_{\rm sle})\Delta_T\\ \label{E:P-T2b}
  s.t. ~&\textstyle\sum\nolimits_{t\in \hat{\cal N}} W_{\max} \log_2(1+\hat g^t \hat p^t) \Delta_T= B,\\
&\hat p^t\geq 0, p^t\leq P_{\max}, {{t\in \hat{\cal N}}\subseteq \hat {\cal T}_{\rm id}}
\end{align}
\end{subequations}
where $\hat p^t$ is obtained to help estimate $N^*$, and $\hat{\cal N}$ is
the index set of scheduled idle time slots, which is a subset of the
estimated idle time slot set $\hat {\cal T}_{\rm id}$.

The solution of problem \eqref{P-T:2-est} can be found by using the similar
way as solving problem \eqref{P-T:2}. The optimal number of scheduled idle
time slots can be estimated as $\hat N^*$, and finally the estimated
Lagrange multiplier $\hat \nu^*$ can be obtained with $\hat N^*$.




\vspace{-1mm}
\subsection{Estimation of Threshold ${g_{\rm th}^*}$}
Recall that $g_{\rm th}^*$ is used to select $N^*$ idle time slots from the
set ${\cal T}_{\rm id}$ with highest equivalent channel gains. Since $\hat
N^*$ and $\hat {\cal T}_{\rm id}$ are inevitably with estimation errors,
intuitively setting the threshold according to the estimated values of $g^t$
and $N^*$ may lead to insufficient number of scheduled idle time slots,
which results in outage in the transmission. In order to control the outage
probability, we find a conservative threshold. To this end, we ensure the
probability that at least $\hat N^*$ idle time slots will be scheduled is
close to one, i.e.,\vspace{-1mm}
\begin{equation}\label{E:outage}
  {\rm Pr}(n\geq \hat N^*,t\in \hat {\cal T}_{\rm id}) = 1-\epsilon,
\end{equation}
where $\epsilon$ is a small value to control the outage probability.

With the help of network and user level context information, we obtain a
closed form expression of the probability in the following proposition,
where the accuracy of the approximation will be evaluated later via
simulations.
\begin{proposition}\label{Pro:1}
When $T$ and $\hat T_{\rm id}$ are large, the probability that at least
$\hat N^*$ idle time slots are scheduled  can be approximated as
follows,\vspace{-2mm}
\begin{equation}\label{E:Pr_real}
\!{\rm Pr}(n\geq \hat N^*,t\in \hat {\cal T}_{\rm id})\! \approx \!\textstyle
\sum\nolimits_{n=\lceil{\hat N^*T}/{\hat T_{\rm id}}\rceil}^{T}
\!\!\textstyle\binom{T}{n}
 (q)^n(1\!-\!q)^{T\!-\!n}
\end{equation}
where $\lceil\cdot\rceil$ is the ceiling function, \vspace{-2mm}
\begin{align}\label{E:q}
\textstyle q=
\frac{1}{T}\sum_{t=1}^T\sum_{i=0}^{N_t-1}\frac{1}{i!}\big({ \frac{G\sigma^2}{\hat \alpha^t} \hat g_{\rm th}}\big)^i
 e^{-{ \frac{G\sigma^2}{\hat \alpha^t}\hat g_{\rm th}}},
\end{align}
and $\hat g_{\rm th}$ is the estimated threshold.
\end{proposition}
\begin{proof}
See Appendix A.
\end{proof}

Since a  higher threshold means that less idle time slots will be scheduled,
the probability ${\rm Pr}(n\geq \hat N^*,t\in \hat {\cal T}_{\rm id})$ is a
decreasing function of the threshold $\hat g_{\rm th}$. Then, by
substituting \eqref{E:Pr_real} into \eqref{E:outage}, $q$ can be found
numerically by bisection searching. Further using \eqref{E:q}, the estimated
threshold $\hat g_{\rm th}^*$ satisfying \eqref{E:outage} can be finally
obtained by bisection searching.

\begin{remark}
The proposed BS sleeping, scheduling and power allocation policy is
implemented as follows. When the request of the NRT user arrives at the 1st
time slot, we can estimate the two parameters $\nu^*$ and ${g_{\rm th}^*}$
with available context information using the method proposed in this
section. Then, when we design the resource allocation in the $t$th time
slot, by using the values of $g^t$, $W^t$ and $p^t_{\rm RT}$ that are
available in current time instant, \emph{PC-algorithm} can be applied to
obtain $p^{t*}$, and the scheduling can be obtained as $m^{t*} = {\rm
Sign}(p^{t*})$.  When $p^{t*} =0$, $t\in{\cal T}_{\rm id}- {\cal N}$,  the
BS goes to sleeping mode.
%
\end{remark}

\section{Simulation}

In this section, we evaluate the proposed context-aware resource allocation
by simulations.

\begin{figure}
\centering
\includegraphics[width=0.35\textwidth]{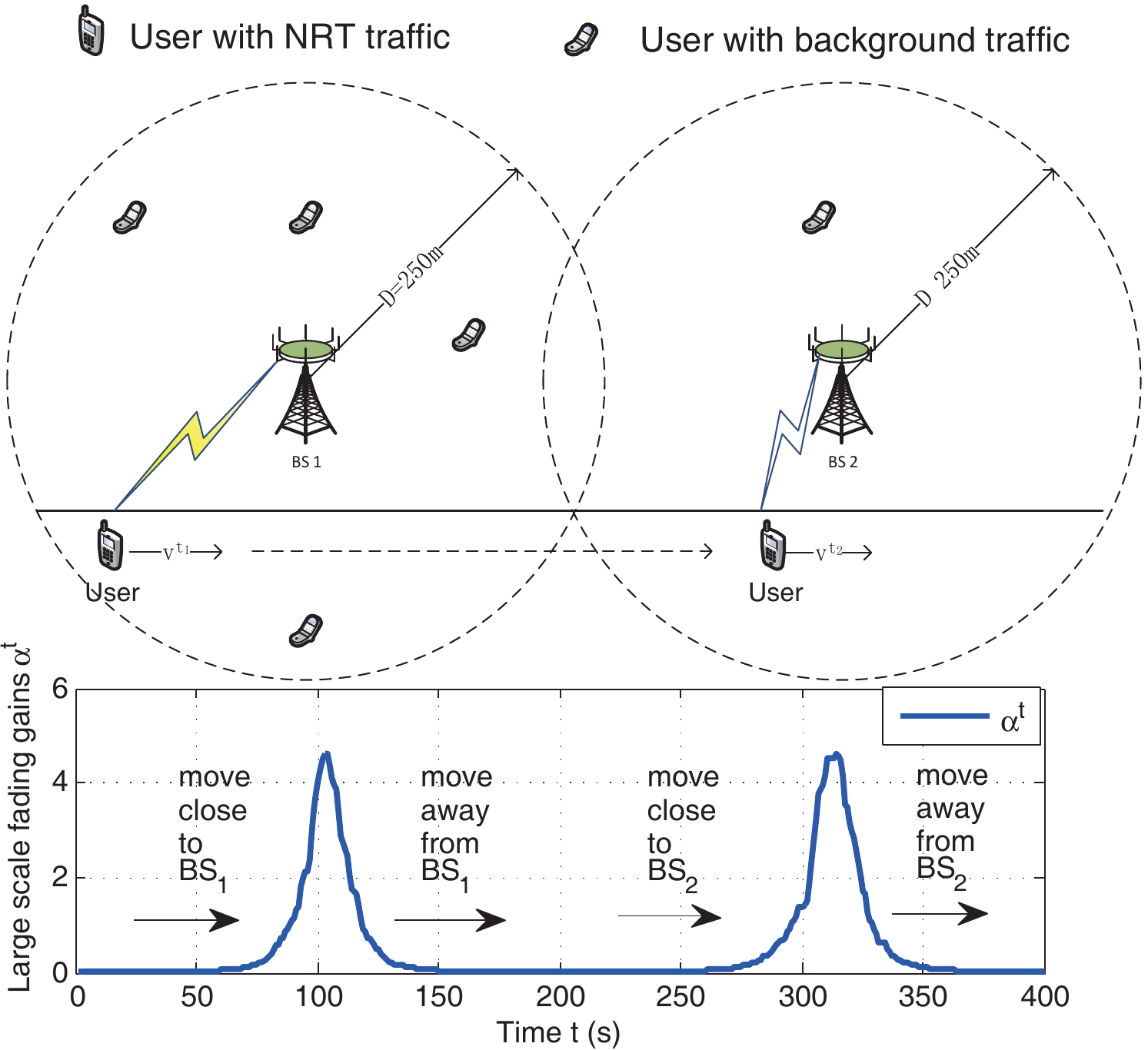}
\caption{Two macro cells with radius $D=250$ m, $N_t=4$, the two BSs alternatively serve a mobile
user with NRT traffic in $T$ time slots each with interval of $\Delta_T=1s$ when there
are random arrival background requests. The maximal transmit
power and bandwidth of each BS are $P_{\max} = 40$ W and $W_{\max} = 10$ MHz. The path loss model is $15.3+37.6\log_{10}(d)$, where $d$ is
the distance between the BS and user in meter \cite{TR36.814}. The
cell-edge  SNR is set as 5 dB, which is the SNR received at the distance $D$
when a BS transmits with $P_{\max}$, reflecting both noise and intercell
interference. The small scale channel is subject to Rayleigh block fading.
The circuit power consumption parameters are $P_{\rm act}$ = $233.2$ W, $P_{\rm sle}$ =$150$ W, and
$\xi=21.3\%$\cite{Auer2011}.
} \label{F:example}
\end{figure}

The simulation setups are shown in Fig. \ref{F:example}. To reflect the
resource utilization of busy BSs, the requests of background traffic is
modeled as Poisson process with average rate $\lambda_1$ and $\lambda_2$ for
the 1st and 2nd BSs, and the service time of each request follows
exponential distribution with mean $V=2$ time slots \cite{das2003framework}.
Each request with background traffic is assumed to occupy the same transmit
power of $8$ W and bandwidth of $2$ MHz. Then, a BS can serve at most
$N_C=5$ requests of background traffic in one time slot, and the newly
arrived request will not be admitted if the BS has been fully occupied.
Based on this request model of the background traffic, the average number of
idle time slots for the $i$th BS during $T$ time slots can be obtained as
$\hat T_{\rm id} = \sum_{i=1}^2{\rm Pr_{id}}{(i)T_i}$, where ${\rm
Pr_{id}}{(i)}= {1}/\big({\sum_{n=0}^{N_C} {({\lambda_i V} )^n}/{n!}}\big)$
obtained in \cite{allen1990probability} is the probability that the $i$th BS
is idle, and $T_i$ is the duration when NRT user is in the $i$th cell. When
$\lambda_i =0.2$, the probability is 0.67, i.e., 67\% time slots are idle.

The user with NRT traffic moves across the two macro cells along a line with
speed $v^t$ uniformly distributed in $(0,5)$ m/s, i.e., during the $t$th
time slot, the user is in constant speed $v^t$ and in each time slot, the
speed is different. Therefore, the average speed during $T$ time slots is
$2.5$ m/s from which we can roughly estimate the large scale fading gains.
$B=5$ GBits of data need to be conveyed within the $T$ time slots. Shannon
Capacity is considered in \eqref{E:Rate}, i.e., $G=0$ dB. The results are
obtained from 100 Monte Carlo trails.

We first validate the approximation in Proposition \ref{Pro:1}. From the
simulated probability ${\rm Pr}(n\geq \hat N^*,t\in\hat {\cal T}_{\rm id})$
with given $\hat N^*$ and $\hat T_{\rm id}$ as well as the numerical results
in Fig. \ref{F:0}, we can see that the approximation is accurate.
\begin{figure}
\centering
\includegraphics[width=0.40\textwidth]{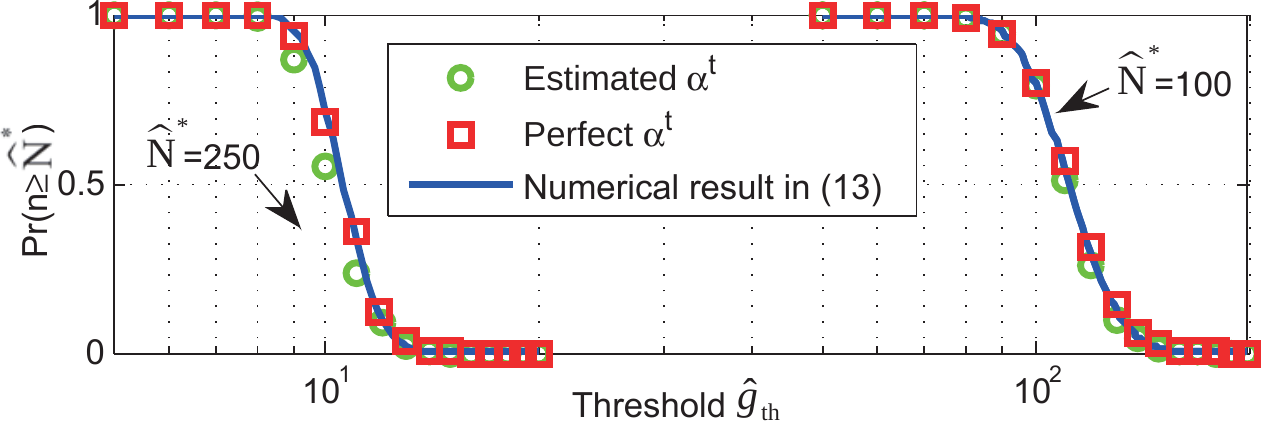}
\caption{The accuracy of the approximation of ${\rm Pr}(n\geq \hat N^*,t\in\hat {\cal T}_{\rm id}) $
when $\alpha^t$ is perfectly known
and is estimated with errors, $T = 1000, \hat T_{\rm id} = 500$.}
\label{F:0}
\end{figure}

To show the impact of different context information on saving energy, the
following approaches are simulated. With all approaches, the BS will turns
      into sleep mode when there is no any traffic to serve.
\begin{enumerate}
  \item \emph{SE-maximizing without context information} (Legend ``SE-No
      context"): Without considering context information, the BS serves
      the NRT traffic as best effort service that maximizes the capacity
      in each time slot, i.e., $p^t= P_{\max}-p^t_{\rm RT}$.
  \item \emph{EE-maximizing without context information} (Legend ``EE-No
      Context"): The BS maximizes the ratio of the instantaneous rate
      transmitted and the power consumed in each time slot.
  \item \emph{Resource allocation with user and application level context
      information} (Legend ``U$\&$A Context "): The BS estimates the
      Lagrange multiplier and threshold without considering the context
      information from network, i.e., without the worst case assumption.

  \item \emph{Resource  allocation with  network, user and application
      level context information} (Legend ``All Context "): This is the
      proposed context-aware resource allocation.

\item \emph{Resource allocation with all future information }(Legend
    ``Upper Bound"): This is the optimal solution in section III with
    perfect future information, which provides the  performance upper
    bound.
\end{enumerate}
\begin{figure}
\centering
\includegraphics[width=0.38\textwidth]{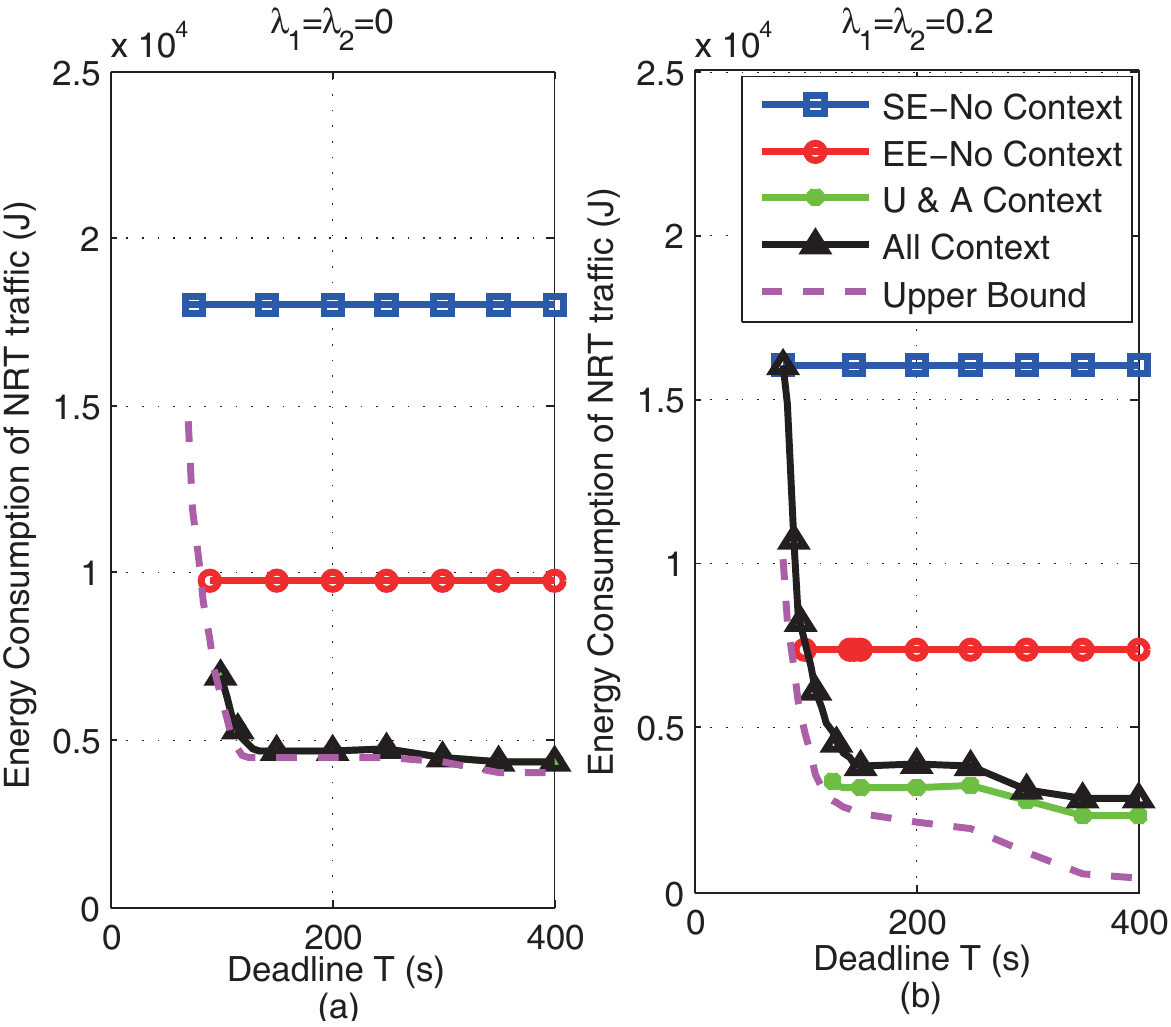}
\caption{Energy consumption of NRT traffic versus
deadline $T$.}
\label{F:1}
\end{figure}

Since context information is exploited to save energy for the NRT traffic,
we only show the results for this kind of traffic, which is the total energy
consumption minus the basic energy consumption, $\sum_{t=1}^T (P^t -
P^t_{\rm B})\Delta_T$.

In Fig. \ref{F:1}, we show the impact of the deadline of NRT traffic on
energy saving, where only the results without outage are presented. We can
see that with the increase of deadline $T$, the SE- and EE-maximizing
approaches consume more energy than the proposed context-aware resource
allocation due to not considering the application level context information.
Moreover, the energy consumption of the proposed method decreases as the
deadline increases. When BSs are without background traffic (i.e., in Fig.
\ref{F:1}(a)), the proposed method almost performs the same as the upper
bound, which implies that the rough knowledge of user level  context
information is good enough to estimate the two parameters, $\nu^*$ and
$g^*_{\rm th}$. When the BSs are with background traffic (i.e., in Fig.
\ref{F:1}(b)), there is a gap between ``All Context" and ``Upper Bound",
which comes from the worst case assumption on the network level context
information when we estimate parameter $\hat \nu^*$. Comparing Fig.
\ref{F:1}(a) with Fig. \ref{F:1}(b), we see that the energy consumption is
lower when the background traffic exists and $T$ is large, because there are
more chances to transmit in the busy time slots, which reduces the circuit
energy consumption for the NRT traffic. ``All Context" consumes  a little
more energy than ``U$\&$A Context". This is because the transmission is more
conservative when considering the network level context information, in
order to ensure lower outage probability.

\begin{figure}
\centering
\includegraphics[width=0.36\textwidth]{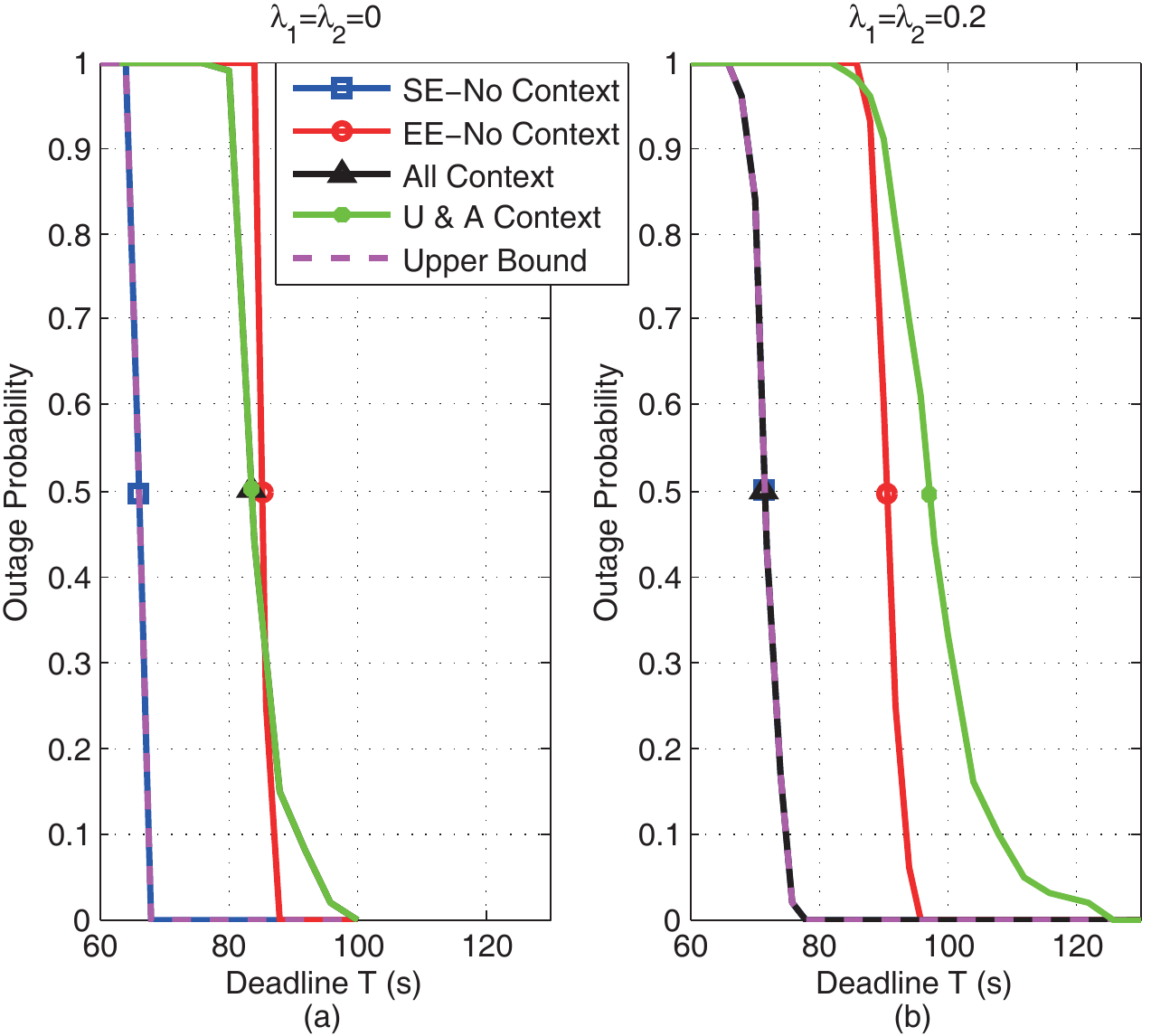}
\caption{Outage probability versus deadline T.}
\label{F:2}
\end{figure}

In Fig. \ref{F:2}, we show the outage probability corresponding to Fig.
\ref{F:1}. When there is no background traffic, the context-aware resource
allocation yields similar outage probability to the EE-maximizing method
without context information. When there exists background traffic as in Fig.
\ref{F:2}(b), the context information from network helps reduce the outage
probability. Note that in Fig. \ref{F:1}(b) the performance exploiting and
not exploiting the context information from network are close. This
indicates that the context information from the application and users helps
save the energy and that from the network helps reduce the outage
probability.

\vspace{-1mm}
\section{Conclusion}
In this paper, we proposed a context-aware resource allocation policy for
the non-real-time traffic  when the BS is occupied by background traffic.
The application, user and network levels context information were exploited
to save energy. Simulation results showed that the application and user
levels context information can help save energy, and the network level
context information can help reduce the outage probability in transmission.
\vspace{-1mm}

\appendices
\section{Proof of Proposition 1}
Because $\hat {\cal T}_{\rm id}$ is estimated as a random guess from ${\cal
T} \triangleq\{1,\ldots,T\}$, $\hat {\cal T}_{\rm id}$ can be viewed as
$\hat T_{\rm id}$ samples from the statistical population $\cal T$.
Consequently, the probability to select $N^*$ samples from $\hat {\cal
T}_{\rm id}$ is equal to the probability to select $\lceil N^*T/\hat T_{\rm
id}\rceil$ samples from $\cal T$ when $T$ and $\hat T_{\rm id}$ approach
infinity. Therefore, when $T$ and $\hat T_{\rm id}$ are large, we
have\vspace{-2mm}
\begin{align}\label{E:App1}
\textstyle\!\!{\rm Pr}(n\geq \hat N^*,t\in\hat{\cal T}_{\rm id}) \!\approx
 \!{\rm Pr}(n\geq\lceil\tfrac{ \hat N^*T}{\hat T_{\rm id}}\rceil,t\in\cal T).
\end{align}

Without loss of generality, the values of equivalent large scale fading
gains and their estimates $ \alpha^t, \hat \alpha^t, t\in\cal T$ are within
a range of $[\alpha_{\min},\alpha_{\max}]$. We divided the range into $J$
intervals as $\delta_1,\ldots,\delta_{J}$, where $\delta_j =
[\alpha_{\min}+(j-1)\delta,\alpha_{\min}+j\delta]$, and $\delta =
(\alpha_{\max}-\alpha_{\min})/J$. The number of $ \hat \alpha^t,t\in \cal T$
located in the interval $\delta_j$ is denoted as $x_j$.

Construct a sequence of i.i.d variables $\beta^t,t\in \cal T$, whose
probability mass function is ${\rm Pr}(\beta^t = \hat \alpha^i) =
\frac{1}{T},i \in \cal T$. Then, the probability that $\beta^t$ is located
in the interval $\delta_j$ is ${\rm Pr}(\beta^t\in\delta_j)=x_j/T$. When $T
\to \infty$, based on Borel's law of large numbers, the average number of $
\beta^t,t\in \cal T$ located in the interval $\delta_j$ is $
 {\mathbb E}\{z_j\} =  \textstyle \lim_{T\to\infty} z_j = {\rm Pr}(\beta^t\in \delta_j) T = x_j
$. Therefore, we can use $\beta^t, t\in\cal T$ to approximate $\hat
\alpha^t, t\in\cal T$. Further considering that $\|{\bf h}^t\|^2,t\in\cal T$
are i.i.d., the probability to select $\lceil\tfrac{ \hat N^*T}{\hat T_{\rm
id}}\rceil$ equivalent channel gains from $g^t = \alpha^t\|{\bf
h}^t\|^2/(G\sigma^2),t\in\cal T$ can be approximated as the probability to
select same number of values from $\beta^t\|{\bf
h}^t\|^2/(G\sigma^2),t\in\cal T$, i.e., \vspace{-2mm}
\begin{align}\nonumber
  {\rm Pr}(n\geq\lceil\tfrac{ \hat N^*T}{\hat T_{\rm
id}}\rceil,t\in{\cal T})\!\! \approx\!\!\textstyle
\sum\nolimits_{n=\lceil{\hat N^*T}/{\hat T_{\rm id}}\rceil}^{T}
\!\!\textstyle\binom{T}{n}
 (q)^n(1-q)^{T-n},
\end{align}
where $q$ is the probability that $\beta^t\|{\bf h}^t\|^2/(G\sigma^2)$
exceeds the threshold $\hat g_{\rm th}$. Then, together with \eqref{E:App1},
we have \eqref{E:Pr_real}.

Because $\|{\bf h}^t\|^2$ follows Gamma distribution $\mathbb{G}(N_t, 1)$
with probability density function $f(h) = \frac{h^{N_t -1}}{\Gamma(N_t)}
e^{-h}$, the probability that $ \beta^t\|{\bf h}^t\|^2/(G\sigma^2) > \hat
g_{\rm th}$ can be derived as
\begin{align}\nonumber
  q =& {\rm Pr}(\tfrac{\beta^t \|{\bf h}^t\|^2}{G\sigma^2}\geq \hat g_{\rm th})
\nonumber  \\
  =&\textstyle
  \sum_{i=1}^T {\rm Pr}(\beta^t
 =\hat \alpha^i ){\rm Pr}(\|{\bf h}^t\|^2\geq { \frac{G\sigma^2}{\hat \alpha^i}\hat g_{\rm th}} )
\nonumber\\=& \textstyle\frac{1}{T}\!\!\sum_{i=1}^T\!\!
\int_{ \frac{G\sigma^2}{\hat \alpha^i}\hat g_{\rm th}}^\infty f(h){\rm d } h
\nonumber  \\= &\textstyle
\frac{1}{T}\!\!\sum_{i=1}^T\!\!\sum_{j=0}^{N_t-1}\!\!\frac{1}{j!}
\big({ \frac{G\sigma^2}{\hat \alpha^i}\hat g_{\rm th}}\big)^j
 e^{-{ \frac{G\sigma^2}{\hat \alpha^i}\hat g_{\rm th}}},\nonumber
\end{align}
which is \eqref{E:q}.
\vspace{-1mm}

\bibliographystyle{IEEEbib}
\bibliography{IEEEabrv,2014YCT}

\begin{thebibliography}{10}
\providecommand{\url}[1]{#1}
\csname url@samestyle\endcsname
\providecommand{\newblock}{\relax}
\providecommand{\bibinfo}[2]{#2}
\providecommand{\BIBentrySTDinterwordspacing}{\spaceskip=0pt\relax}
\providecommand{\BIBentryALTinterwordstretchfactor}{4}
\providecommand{\BIBentryALTinterwordspacing}{\spaceskip=\fontdimen2\font plus
\BIBentryALTinterwordstretchfactor\fontdimen3\font minus
  \fontdimen4\font\relax}
\providecommand{\BIBforeignlanguage}[2]{{%
\expandafter\ifx\csname l@#1\endcsname\relax
\typeout{** WARNING: IEEEtran.bst: No hyphenation pattern has been}%
\typeout{** loaded for the language `#1'. Using the pattern for}%
\typeout{** the default language instead.}%
\else
\language=\csname l@#1\endcsname
\fi
#2}}
\providecommand{\BIBdecl}{\relax}
\BIBdecl

\bibitem{Choongul2012concept}
C.~Park, Y.~Seo, K.~Park, and Y.~Lee, ``The concept and realization of
  context-based content delivery of {NGSON},'' \emph{IEEE Commun. Mag.},
  vol.~50, no.~1, pp. 74--81, Jan. 2012.

\bibitem{Abou2013Predictive}
H.~Abou-zeid and H.~Hassanein, ``Predictive green wireless access: exploiting
  mobility and application information,'' \emph{IEEE Wireless Commun.},
  vol.~20, no.~5, pp. 92--99, Oct. 2013.

\bibitem{abou2014toward}
H.~Abou-Zeid and H.~S. Hassanein, ``Toward green media delivery: location-aware
  opportunities and approaches,'' \emph{IEEE Wireless Commun.}, vol.~21, no.~4,
  pp. 38--46, Aug. 2014.

\bibitem{Draxler2014Anticipatory}
M.~Dr{\"a}xler, P.~Dreimann, and H.~Karl, ``Anticipatory power cycling of
  mobile network equipment for high demand multimedia traffic,'' in \emph{Proc.
  IEEE GREENCOM}, 2014.

\bibitem{das2003framework}
S.~K. Das, S.~K. Sen, K.~Basu, and H.~Lin, ``A framework for bandwidth
  degradation and call admission control schemes for multiclass traffic in
  next-generation wireless networks,'' \emph{IEEE J. Select. Areas Commun.},
  vol.~21, no.~10, pp. 1790--1802, Dec. 2003.

\bibitem{cioffi1991multicarrier}
J.~M. Cioffi, ``A multicarrier primer,'' \emph{ANSI T1E1}, vol.~4, pp. 91--157,
  1991.

\bibitem{Auer2011}
G.~Auer, V.~Giannini, C.~Desset, and e.~I.~Godor, ``How much energy is needed
  to run a wireless network?'' \emph{IEEE Wireless Commun.}, vol.~18, no.~5,
  pp. 40--49, Oct. 2011.

\bibitem{froehlich2008route}
J.~Froehlich and J.~Krumm, ``Route prediction from trip observations,'' Soc.
  Automotive Eng. World Congress, Tech. Rep., 2008.

\bibitem{TR36.814}
{TR 36.814 V1.2.0}, ``{F}urther {A}dvancements for {E-UTRA} {P}hysical {L}ayer
  {A}spects ({R}elease 9),'' \emph{3GPP}, June 2009.

\bibitem{allen1990probability}
A.~O. Allen, \emph{Probability, statistics, and queueing theory: with computer
  science applications}.\hskip 1em plus 0.5em minus 0.4em\relax Gulf
  Professional Publishing, 1990.

\end{thebibliography}
\end{document}